\documentclass[twoside,11pt,reqno]{amsart}
\usepackage[T1]{fontenc}
\usepackage{txfonts}
\DeclareSymbolFont{symbols}{OMS}{cmsy}{m}{n}

\usepackage[margin=1.2in]{geometry}
\usepackage[utf8]{inputenc}

\usepackage{mathrsfs}
\usepackage{mathtools}
\usepackage{amsfonts}
\usepackage{amsmath,amssymb,amsthm}
\usepackage{enumerate}
\usepackage{hyperref}
\usepackage[alphabetic,initials]{amsrefs}

\newtheorem{thm}{Theorem}[section]

\newtheorem{prop}[thm]{Proposition}
\theoremstyle{definition}

\newtheorem{question}[thm]{Question}
\newtheorem{problem}[thm]{Problem}
\newtheorem{defi}[thm]{Definition}
\theoremstyle{remark}
\newtheorem{rmk}[thm]{Remark}

\newcommand{\Hil}{\mathcal{H}}
\newcommand{\slot}{{~\cdot~}}

\newcommand{\SU}{\mathop{\mathsf{SU}}}

\newcommand{\Mob}{\mathsf{M\ddot ob}}

\newcommand{\Sc}[1][]{\mathbb{S}^{1#1}}

\newcommand{\cF}{\mathcal{F}}

\newcommand{\cB}{\mathcal{B}}

\newcommand{\cD}{\mathcal{D}}

\newcommand{\cC}{\mathcal{C}}

\newcommand{\A}{\mathcal{A}}
\newcommand{\cI}{\mathcal{I}}

\newcommand{\RR}{\mathbb{R}}

\newcommand{\CC}{\mathbb{C}}

\newcommand{\ZZ}{\mathbb{Z}}
\newcommand{\NN}{\mathbb{N}}

\DeclareMathOperator{\Diff}{Diff}
\DeclareMathOperator{\End}{End}

\DeclareMathOperator{\Rep}{Rep}

\DeclareMathOperator{\Hom}{Hom}

\DeclareMathOperator{\id}{id}

\DeclareMathOperator{\B}{B}

\newcommand{\e}{\mathrm{e}}

\newcommand{\ima}{\mathrm{i}}

\newcommand{\op}{\mathrm{op}}

\usepackage{xspace}
\DeclareRobustCommand{\eg}{e.g.\@\xspace}
\DeclareRobustCommand{\cf}{cf.\@\xspace}
\DeclareRobustCommand{\ie}{i.e.\@\xspace}
\makeatletter
\DeclareRobustCommand{\etc}{%
    \@ifnextchar{.}%
        {etc}%
        {etc.\@\xspace}%
}
\makeatother

\DeclareMathOperator{\Irr}{Irr}

\newcommand{\cG}{\mathcal{G}}
\newcommand{\bim}[4][]{{}\prescript{\vphantom{#1}}{#2}{#3}^{#1}_{#4}}

\newcommand{\rev}[1]{{#1}^\mathrm{rev}}

\begin{document}
\date{\today}
\dateposted{\today}
\newcommand{\mytitle}{
The Relation between Subfactors arising from Conformal Nets and the Realization of Quantum Doubles 
}
\title{\mytitle}
\address{Vanderbilt University, Department of Mathematics, 1326 Stevenson
Center, Nashville, TN 37240, USA}
\author{Marcel Bischoff}
\email{marcel.bischoff@vanderbilt.edu}
\thanks{Supported in part by NSF Grant DMS-1362138}
\begin{abstract}
We give a precise definition for when a subfactor arises 
from a conformal net which can be motivated 
by classification of defects.
We show that a subfactor $N\subset M$ arises from a conformal net if there is a conformal net 
whose representation category is the quantum double of $N\subset M$.
\end{abstract}

\maketitle

\setcounter{tocdepth}{2}
\tableofcontents

\section{Introduction}
Finite index subfactors $N\subset M$ generalize finite group fixed points and can be seen as 
describing quantum symmetries.
An important invariant of a subfactor $N\subset M$ is its index $[M:N]$. 
It generalizes the index of a subgroup in the sense that for group-subgroup subfactors we have $[M^H:M^G]=[G:H]$.
By Jones' index theorem \cite{Jo1983} the index takes values in:
$$
  [M:N]\in \left\{ 4\cos^2\left(\frac{\pi}{m}\right) : m=3,4,\ldots\right\}\cup [4:\infty]
$$
and all subfactors with index less than four have finite depth. For index greater than 4 there are some 
known exotic subfators with finite depth and the classification has been recently pushed to $5\frac14$ \cite{JoMoSn2014,AfMoPe2015}.
Finite depth subfactors are rather algebraic objects, but since everything is defined on a Hilbert space
this algebraic structure still has important positivity properties. It is an interesting question if and how they arise 
describing symmetries in models of quantum physics.

Using the Haag--Kastler axioms of algebraic quantum field theory \cite{Ha} one can describe quantum field theory directly 
using nets of local von Neumann algebras. Under natural assumptions the local algebras turn out to be factors and are in many cases isomorphic to the hyperfinite type III${}_1$ factor \cite{Co1973,Ha1987}. The theory of Doplicher--Haag--Roberts 
superselection sectors studies the representation theory of Haag--Kaster nets in terms of so-called localized endomorphisms.
Each endomorphism gives a subfactor, but in higher dimensional QFT the index is rather boring and takes values in $\{n^2:n\in \NN\}$, indeed all subfactors come from a representation of a compact group \cite{DoRo1989}.
On the other hand, in low-dimensional QFT the superselection theory gets interesting. The superselection sectors have braid group statistics \cite{FrReSc1989}
and the index is in general not a square of an integer. For example all values $\left\{ 4\cos^2\left(\frac{\pi}{m}\right) : m=3,4,\ldots\right\}$ 
of the index can be realized by a loop group model $\A_{\SU(2)_{m-2}}$ of $\SU(2)$ at level $m-2$ \cite{Wa}.
But these subfactors always come with a braiding and there are known subfactors which do not admit such a braiding.
Now there are two ways out. First: The quantum double of a subfactor gives a braided subfactor, namely it gives 
a unitary modular tensor category
and one can try to construct quantum field theories realizing such quantum doubles as DHR category of superselection sectors. But these
does in general not give the original subfactor back (at least not directly). Second:
One can look into higher structures of the quantum field theory for example one can allow boundary conditions and defects.
Here one needs to consider extensions and new subfactors arise which are not braided.
The goal of this note is to show that these two directions are directly related.

The experience shows that it might be enough to consider completely rational M\"obius covariant nets on the circle
which describes chiral conformal field theory (CFT).\footnote{Probably one wants assume diffeomorphism covariance. But one the one hand
for our structural results this is not necessary to assume. On the other hand, we are not aware that there is a known completely rational net 
which is not diffeomorphism covariant.} In this framework we want to make a precise 
statement what is meant by 
 the following question.

\begin{question}[\cf \cite{Jo2014}] Do all subfactors come from quantum field theory?
\end{question}
By a subfactor we mean from now on always a finite index, finite depth subfactor 
which is 
hyperfinite of type III${}_1$.
If we have finite index and finite depth subfactor $N\subset M$ of type  
II${}_1$ we can always pass to hyperfinite III${}_1$.
  $\tilde N:=N\subset A \subset \tilde M :=M\otimes A$ with 
$A$ the hyperfinite type III${}_1$ factor with the same standard invariant.

Kawahigashi, Longo and M\"uger \cite{KaLoMg2001} showed that under the rather natural assumption 
of complete rationality of a conformal net $\A$, its representation 
category $\Rep(\A)$ is a \textbf{unitary modular tensor category (UMTC)}.
Unitary modular tensor categories play a prominent role in topological quantum computing,
they give 3 manifold invariants and topological quantum field theories
 via the Reshetekin--Turaev construction
\cite{ReTu1991,Tu1994}.
The natural question arises if one can find a general solution to the following problem (\cf \cite{Ka2015-02}).
\begin{problem}
\label{prob:UMTC}
 Given a unitary modular tensor category $\cC$, find a conformal net 
  with $\Rep(\A)$ braided equivalent to $\cC$.
\end{problem}
The modular tensor category encodes topological information about 
the conformal net in terms of a three dimensional topological field theory.
Note that nevertheless the conformal net has more information than just its representation category,
because there infinitely many non-isomorphic conformal nets sharing the same UMTC as representation category. They arise
by tensoring with a holomorphic net (see below).

Nevertheless, the problem does not seem completely hopeless. 
It is similar to inverse scattering problems in quantum field theory.
One idea is 
to use all or some of the data of the UMTC to construct a statistical mechanical model
which in a limit at critical temperature gives a conformal field theory. 
This way is full of technical difficulties and we will not further comment on it.

Since the quantum double $\cD(N\subset M)$ of a finite index finite depth subfactor $N\subset M$
is a UMTC, it seems to be natural to 
consider the following subproblem of Problem \ref{prob:UMTC} (\cf \cite{EvGa2011,Ka2015-02}).
\begin{problem} 
\label{prob:QDnet}
Given a finite index finite depth subfactor $N\subset M$, 
\begin{enumerate}
  \item Find a conformal 
  net $\A$, such that $\Rep(\A) \cong \cD(N\subset M)$.
  \item Find a conformal net $\A$,  such that $\Rep(\A) \supset \cD(N\subset M)$,
  \ie $\cD(N\subset M)$ is equivalent to a full subcategory of $\Rep(\A)$.
  In this case $\Rep(\A)\cong \cD(N\subset M)\boxtimes \cC$, where $\cC$ is a modular tensor category.
\end{enumerate}
\end{problem}

If we have a UMTC $\cC$ we get a UMTC $\rev{\cC}$
by replacing the braiding with the opposite braiding. In general $\rev{\cC}$ is not braided equivalent to $\cC$.
Since the braiding is instrinisically defined and conformal nets have a posititvity of  energy condition, there seems to be no easy
way to get a net realizing the opposite braiding without destroying the positivity of energy condition. Therefore the following question
naturally arises:
\begin{question}[\cite{Lo2012p}] Let $\A$ be a completely rational net. Is there a completely rational net 
  $\tilde \A$ with $\Rep(\tilde\A) \cong \rev{\Rep(\A)}$?
\end{question}
This question can be answered positively, if we solve the following problem (see Proposition \ref{prop:Cofinite}).
\begin{problem}
Given a completely rational net $\A$, find a holomorphic net $\cB$,
such that $\A\subset \cB$ is normal and cofinite.
\end{problem}

Motivated by the study of phase boundaries and topological defects, we say a subfactor $N\subset M$ arises from a conformal net $\A$,
if there are two relatively local extensions $\cB_a,\cB_b\supset \A$,
and a sector $\beta\colon \cB_a(I)\to\cB_b(I)$ related to $\A$, such that $N\subset M\approx \beta(\cB_a(I))\subset \cB_b(I)$,
see Definition \ref{defi:Subfactor}
\begin{prop} 
Let $N\subset M$ be a finite index, finite depth subfactor.
\begin{enumerate}
  \item If there is a conformal net with $\Rep(\A)$ braided equivalent to $\cD(N\subset M)$, 
    then $N\subset M$ arises from $\A$. 
Actually, it is enough that $\Rep(\A)$ contains a full subcategory braided equivalent to $\cD(N\subset M)$.
  \item Conversely, if $N\subset M$ arises from $\A$, then there is 
  a 2D conformal net $\cB_2\supset \A\otimes \A$ with $\Rep(\cB_2)\cong \cD(N\subset M)$.
  \item Further, if $N\subset M$ arises from $\A$, and there is a net $\tilde A$, with $\Rep(\tilde \A)\cong \rev{\Rep(\A)}$
    then there is a conformal net $\cB\supset \A\otimes \tilde \A$ with $\Rep(\cB)$ braided equivalent to $\cD(N\subset M)$.
\end{enumerate}
\end{prop}

\section{Subfactors and Unitary Fusion Categories}
Let $M$ be the hyperfinite type III${}_1$ factor
and $N\subset M$ a finite index and finite depth subfactor.
We denote by $\iota\colon N\to M$ the canonical inclusion map,
which is a morphism (normal $\ast$-homomorphism) $N\to M$.
Then finite index of $N\subset M$ is equivalent with the existence of a morphism $\bar\iota\colon M \to N$, such that 
$\id_N\prec \bar\iota\circ \iota$
and $\id_M\prec \iota\circ\bar\iota$ \cf \cite{Lo1990}.
Here we say a morphism $\rho\colon N\to M$ is contained in $\sigma\colon N\to M$, written $\rho\prec\sigma$ if and only if there is an isometry 
$e\in \Hom(\rho,\sigma)=\{t\in M: t\rho(N) =\sigma(N) t\}$.

All endomorphisms $\rho$ of $M$, such that $\rho(M)\subset M$ have finite index,
form a rigid C${}^\ast$-tensor category $\End_0(M)$. The arrows $t\colon \rho\to \sigma$ are given by $t\in \Hom(\rho,\sigma)$ as above
and the tensor product is given by composition of endomorphisms. An endomorphism $\rho$
is called irreducible if $\Hom(\rho,\rho)=\CC \cdot 1$ and since $\Hom(\rho,\rho)=\rho(N)'\cap N$ this is equivalent
with the subfactor $\rho(N)\subset N$ being irreducible.
We denote by $[\rho]$ the sector of $\rho$ which is the unitary equivalence class $\{u\rho(\slot)u^\ast : u\in N \text{ unitary}\}$.
There is a direct sum, well-defined on sectors, given by $\rho\oplus\sigma=v_1\rho(\slot)v_1^\ast+v_2\sigma(\slot)v_2^\ast$ with $v_i$ isometries 
fulfilling the Cuntz algebra relations: $\sum_{i}v_iv_i^\ast =1$
and $v_i^\ast v_j=\delta_{i,j}$.

A finite index subfactor $N\subset M$ with $\iota\colon N \to M$ and conjugate $\bar\iota\colon M\to N$ 
gives two rigid C${}^\ast$-tensor categories,
\begin{itemize}
  \item
the dual even part 
$\bim[N\subset M]  N\cF N=\langle\rho\prec (\bar\iota\circ \iota)^n\rangle \subset \End_0(N)$ and
  \item 
the even part
$\bim[N\subset M] M \cF M =\langle\rho\prec (\iota\circ \bar\iota)^n\rangle \subset \End_0(M)$. 
\end{itemize}
Here the $\langle S \rangle$ denotes the full and replete tensor subcategory
generated by $S$ and closed under taking direct sums.

The subfactor $N\subset M$ is called \textbf{finite depth} if and only if the set $\Irr(\bim[N\subset M]A\cF A)=\{[\rho] :\rho\in\bim[N\subset M]A\cF A \text{ irreducible}\}$
with $A=N,M$ is finite. In this case, $\bim[N\subset M] N\cF N$ and $\bim[N\subset M] M\cF M$ are unitary fusion categories.

The subfactor actually generates a 2-category $\cF_{N\subset M}$ with zero objects $\{N,M\}$, by taking all finite morphisms
$\beta\colon \{N,M\}\to\{N,M\}$ contained in compositions of $\{\iota,\bar\iota\}$, such that 
$\bim[N\subset M] N\cF N$ and $\bim[N\subset M] M \cF M$ sit in the $N-N$ and $M-M$ corner, respectively, of $\cF_{N\subset M}$.

We remark that a unitary fusion category given as a full and replete subcategory $\bim M\cF M\subset \End_0(M)$ is completely fixed by its finite set of sectors. Conversely, given
a finite set of endomorphisms 
$\Delta=\{\rho_0=\id,\rho_1,\ldots,\rho_n\}$
which is closed under
\begin{itemize}
  \item conjugates, \ie
there is a permutation $i\mapsto \bar i$ on $\{1,\ldots, n\}$, such that $[\rho_{\bar i}]=[\bar\rho_i]$
  and
  \item fusion, \ie there are numbers $N^k_{ij}$, the so-called fusion rule coefficients, 
such that $[\rho_i\circ\rho_j]=\bigoplus_{\rho_k\in \Delta} N_{i,j}^{k} [\rho_k]$
\end{itemize}
there is a unique unitary fusion $\langle \rho\in\Delta\rangle \subset \End_0(M)$.

Every fusion category $\bim M\cF M$ can be seen as the even part of a subfactor $N\subset M$.
For example we can simply take the subfactor $N=\rho_\oplus(M)\subset M$,
where $[\rho_\oplus]=\bigoplus_{\rho\in\Irr(\bim M\cF M)}[\rho] $.
This particular subfactor actually has the special feature 
that the depth is two, which implies that there is a weak Kac algebra $Q$, such that 
$N=M^Q\subset M$ \cite{Re1996,NiSzWi1998,NiVa2000}. In this sense, one can see as fusion  categories as representation categories 
of weak Kac algebras, but the choice of $Q$ with this property is not unique.

Further every abstract unitary fusion category $\cF$ can be realized as a (as a concrete fusion category) in $\End_0(M)$,
 \ie there is a full and replete $\bim M \cF M \subset \End_0(M)$, which is equivalent 
to $\cF$. Namely, \cite{HaYa2000}
 gives a realization as bimodules of the hyperfine II${}_1$ factor $R$ and by tensoring with the hyperfinite type III${}_1$
factor we get it realized as endomorphisms (\cf \cite{Lo1990}). Using Popa's theorem \cite{Po1993}
such a realization is unique, \ie if we have another relation on $\tilde M$ then there is an isomorphism $\phi\colon M\to \tilde M$ 
which gives an equivalence of categories (\cf \cite[Proof of Corollary 35]{KaLoMg2001}).

Often one wants a spherical structure on a fusion category. In the unitary case, we don't need to worry, because there is always (a unique up unitary to equivalence) spherical structure \cite{LoRo1997}.

The categorical dimension $d\rho$ of $\rho\in\End_0$ coincides  with the square root the minimal index $[M: \rho(M)]$.
A unitary fusion category $\cF$ is called \textbf{braided} if there is a natural family of unitaries $\varepsilon(\rho,\sigma)\in
\Hom(\rho\sigma,\sigma\rho)$. It is called unitary, \ie \textbf{unitary modular tensor category (UMTC)} if 
$\varepsilon(\rho,\sigma)\varepsilon(\sigma,\rho)=1=1_{\sigma\circ\rho}$ for all $\rho\in\cF$ implies 
$[\sigma]=N[\id]$.

One source of UMTCs are \textbf{quantum doubles of subfactors}.
Given $N\subset M$, we take $S=M\otimes M^\op$ and $\bim S\cC S =\langle \rho\otimes\sigma^\op:\rho,\sigma\in \bim[N\subset M] M \cF M\rangle$ 
which equivalent with the fusion category $\bim[N\subset M] M \cF M \boxtimes (\bim[N\subset M]M\cF M)^\op$.
Then there is a subfactor $\iota_{T\subset S} (T)\subset S$ and 
$\bim T \cC T :=\langle \rho\prec \bar\iota_{T\subset S} \beta\iota_{T\subset S} :\beta \in\bim S\cC S\rangle$
can be identified with the category $Z(\bim[N\subset M] M \cF M)$, which is a UMTC by \cite{Mg2003II}.
Here $Z(\cF) $ denotes the \textbf{unitary (Drinfeld) center} of a unitary fusion category $\cF$.

We could have started with $N$ and $\bim[N\subset M] N\cF N$ and obtain $Z(\bim [N\subset M]N\cF N)$ which is braided equivalent with $Z(\bim M\cF M)$.
Indeed, $\bim[N\subset M] N\cF N$ and $\bim [N\subset M]M\cF M$  are (weakly) Morita equivalent and Morita equivalent fusion categories
have braided equivalent Drinfeld centers by combining \cite{Sc2001} with \cite{Mg2003}.
Therefore we denote the obtained UMTC $Z(\bim[N\subset M] M\cF M)$ by $\cD(N\subset M)$ and call it \textbf{the quantum double} of $N\subset M$.
There is a Galois correspondence between intermediate subfactors $S\subset A \subset T$ and subfusion categories $\cG\subset \bim M \cC M$ \cite{Iz2000}.

\section{Conformal and Completely Rational Nets}
A conformal net is a mathematical prescription of a chiral conformal quantum field theory on the circle using operator 
algebras. A well-behaving family of conformal nets are the so-called completely rational nets, which have 
a representation theory similar to finite groups and quantum groups at root of unity.

We denote by $\cI$ the set of proper intervals $I\subset \Sc$ on the circle and 
by $I'=\Sc\setminus \overline{I}$ the opposite interval.
\label{sec:CN}
By a conformal net $\A$, we mean a local M\"obius covariant net on the circle. 
It associates with every 
proper interval $I\in\cI$ a von Neumann algebra $\A(I)\subset \B(\Hil_\A)$
on a fixed Hilbert space $\Hil$, such that the following properties hold:
    \begin{enumerate}[{\bf A.}]
        \item \textbf{Isotony.} $I_1\subset I_2$ implies 
            $\A(I_1)\subset \A(I_2)$.
        \item \textbf{Locality.} $I_1  \cap I_2 = \emptyset$ implies 
            $[\A(I_1),\A(I_2)]=\{0\}$.
        \item \textbf{Möbius covariance.} There is a unitary representation
            $U$ of $\Mob$ on $\Hil$ such that 
            $  U(g)\A(I)U(g)^\ast = \A(gI)$.
        \item \textbf{Positivity of energy.} $U$ is a positive energy 
            representation, i.e. the generator $L_0$ (conformal Hamiltonian) 
            of the rotation subgroup $U(z\mapsto \e^{\ima \theta}z)=\e^{\ima \theta L_0}$
            has positive spectrum.
        \item \textbf{Vacuum.} There is a (up to phase) unique rotation 
            invariant unit vector $\Omega \in \Hil$ which is 
            cyclic for the von Neumann algebra $\A:=\bigvee_{I\in\cI} \A(I)$.
    \end{enumerate}
A conformal net $\A$ is called \textbf{completely 
rational} if it 
\begin{enumerate}[{\bf A.}]
  \setcounter{enumi}{5}
  \item fulfills the 
    \textbf{split property}, \ie 
    for $I_0,I\in \cI$ with $\overline{I_0}\subset I$ the inclusion 
    $\A(I_0) \subset \A(I)$ is a split inclusion, namely there exists an 
    intermediate type I factor $M$, such that $\A(I_0) \subset M \subset \A(I)$.
  \item is 
 \textbf{strongly additive}, \ie for $I_1,I_2 \in \cI$ two adjacent intervals
obtained by removing a single point from an interval $I\in\cI$ 
the equality $\A(I_1) \vee \A(I_2) =\A(I)$ holds.
  \item for $I_1,I_3 \in \cI$ two intervals with disjoint closure and 
    $I_2,I_4\in\cI$  the two components of $(I_1\cup I_3)'$, the 
    \textbf{$\mu$-index} of $\A$
    \begin{equation*}
      \mu(\A):= [(\A(I_2) \vee \A(I_4))': \A(I_1)\vee \A(I_3) ]
    \end{equation*}
 is finite.
\end{enumerate}
All known examples of completely rational nets also turn out to be covariant with respect to a projective representation of the  diffeomorphism group
of the circle and this leads to representation of the Virasoro algebra, but we just assume M\"obius covariance, although 
the term conformal net often refers to diffeomorphism covariant nets.

Examples of completely rational 
nets are:
\begin{itemize}
  \item Diffeomorphism covariant nets with central charge $c<1$
    \cite{KaLo2004}.
  \item The nets $\A_L$ where $L$ is a positive even lattice
    \cite{DoXu2006} which contain as a special case \cite{Bi2012}
    loop group nets $\A_{G,1}$ at level 1 for $G$ a 
    compact connected, simply connected  simply-laced Lie group.
  \item The loop group nets $\A_{\SU(n),\ell}$ for $\SU(n)$ at level $\ell$.
    \cite{Xu2000}.
\end{itemize}
Further examples of rational conformal nets come from standard constructions:
\begin{itemize}
  \item 
Finite index extensions and subnets of completely rational 
nets. Namely, let $\A\subset \cB$ be a finite subnet \ie $[\cB(I):\A(I)]<\infty$ for
some (then all) $I\in\cI$, then $\A$ is completely rational iff $\cB$ is
completely rational \cite{Lo2003}, in particular orbifolds $\A^G$ of completely
rational nets $\A$ with $G$ a
finite group are completely rational.
  \item Let $\A\subset\cB$ be a co-finite subnet, \ie
$[\cB(I),\A(I)\vee\A^\mathrm{c}(I)] <\infty$ for some (then all) $I\in\cI$,
where the \textbf{coset net} $\A^\mathrm{c}$ is defined by
$\A^\mathrm{c}(I)=\A'\cap\cB(I)$ with $\A'=(\vee_{I\in\cI}\A(I))'$. Then
$\cB$ is completely rational iff $\A$ and $\A^\mathrm{c}$ are completely
rational \cite{Lo2003}.
\end{itemize}

A \textbf{representation} $\pi$ of $\A$ is a family of unital representations 
$\pi=\{\pi_I\colon\A(I)\to \B(\Hil_\pi)\}_{I\in\cI}$ on a common Hilbert space $\Hil_\pi$ which are compatible, i.e.\  $\pi_J\restriction \A(I) =\pi_I$ for $I\subset J$. An example is the trivial representation $\pi_0=\{\pi_{0,I}=\id_{\A(I)}\}$ on 
the defining Hilbert space $\Hil$.
Let us fix through out an abitrary interval $I\in \cI$.
Every 
representation $\pi$ with $\Hil_\pi$ separable 
turns out to be equivalent to a representation \textbf{localized} in $I$, \ie 
$\rho$ on $\Hil$, such that $\rho_J=\id_{\A(J)}$ for $J\cap I=\emptyset$. Then Haag duality implies
that $\rho=\rho_{I}$ is an endomorphism of $\A(I)$.
The \textbf{statistical dimension} of a localized endomorphism $\rho$ is given by
$d\rho=[N:\rho(N)]^{\frac12}$ and we will restrict to endomorphisms with finite statistical dimension. 

The category $\Rep^I(\A)$ of representations of $\A$ with finite statistical dimension which are localized in $I$ naturally is  
a full and replete subcategory of the rigid C$^\ast$ tensor category of endomorphisms $\End_0(\A(I))$.
In particular, this gives the representations of $\A$ the structure of a tensor category \cite{DoHaRo1971}. It has a natural \textbf{braiding}, which is completely fixed
by asking that if $\rho$ is localized in $I_1$ and $\sigma$ in $I_2$ where $I_1$ is left of $I_2$ inside $I$ then $\varepsilon(\rho,\sigma)=1$
 \cite{FrReSc1989}.

\begin{prop}[\cite{KaLoMg2001}] Let $\A$ be completely rational net, then $\Rep^I(\A)$ is a UMTC
and $\mu_\A=\dim(\Rep^I(\A))$, where $\dim(\cC)=\sum_{\rho\in\Irr(\cC)}(d\rho)^2$ is the \textbf{global dimension}.
\end{prop}

We call a completely rational net $\A$ with $\mu(\A)=1$ a \textbf{holomorphic net}. This means that every
representation is equivalent to a direct set of the trivial representation $\pi_0$. 
Examples of holomorphic nets are conformal nets $\A_L$ associated with even lattices $L$ constructed in \cite{DoXu2006}, the
the Moonshine net $\A_\natural$ \cite{KaLo2006} and certain framed nets 
    \cite{KaSu2012}.

Similar to the concept of subgroups, there is the notion of a subnet.
We write $\A\subset \cB$ or $\cB\supset \A$
if there is a representation $\pi=\{\pi_I\colon\A(I)\to\cB(I)\subset\B(\Hil_\cB)\}$ of $\A$ on 
$\Hil_\cB$  and an isometry $V\colon \Hil_\A\to \Hil_\cB$
with $V\Omega_\A=\Omega_\cB$ and $VU_\A(g)=U_\cB(g)V$.
We ask that further that
$Va=\pi_I(a)V$ for $I\in\cI$, $a\in\A(I)$. Define $p$ the projection 
on $\Hil_{\A_0}=\overline{\pi_I(\A(I))\Omega}$. Then $pV$ is a unitary equivalence
of the nets $\A$ on $\Hil_\A$ and  $\A_0$ defined by $\A_0(I)=\pi_I(\A(I))p$ on $\Hil_{\A_0}$.

\section{Subfactors arising from Conformal Nets}
If we have a completely rational net, then 
$A=\A(I)$ is the hyperfinite type III${}_1$ factor.
With $\bim A \cC A =\Rep^I(\A)\subset \End(A)$, 
every $\rho\in\bim A\cC A$ gives a subfactor $\rho(A)\subset A$.

But we are interested in taking all subfactors arising from
$\bim A\cC A$ and Morita equivalent fusion categories.

The philosophy is similar to the one if, that if we have one 
fusion category, \eg the even part of a subfactor,
we look into all Morita equivalent fusion categories and 
then into all subfactors arising this way \cite{GrSn2012}.
This is very much related to Ocneanu's maximal atlas \cite{Oc2001}
the Brauer--Picard groupoid 
of a fusion category \cite{EtNiOs2005}. 

An irreducible finite index overfactor $B\supset \iota(A)$ with $A=\A(I)$, where the dual canonical endomorphism $\bar\iota\circ\iota$ is in $\Rep^I(\A)$
gives rise to a relatively local extension $\cB\supset \A$ and all such extension arise 
in this way. The net $\cB$ is itself is in general not local but just relatively local to $\A$.
The net $\cB$ is local and therefore itself a completely rational net 
\cite{Lo2003}
if the extension comes from a commutative Q-system. 
Relatively local extensions arose by the study of boundary conformal field theory 
\cite{LoRe2004}; they give a boundary net by holographic projection.
By removing the boundary \cite{LoRe2009} one obtains a full CFT on Minkowski 
space (see below).

This motivates the following definition.
\begin{defi} 
\label{defi:Subfactor}
Let $\A$ be a completely rational net.
We say a \textbf{subfactor arises from $\A$} if
it is of the form $\beta(B_a)\subset B_b$, 
where 
$\beta\in \bim {B_b}\cC{B_a}=\langle \beta\prec \iota_b \rho \bar\iota_a :
\rho \in \Rep^I(\A) \rangle$,
with $A=\A(I)$, $B_\bullet=\cB_\bullet(I)$
and $\cB_a,\cB_b\supset \A$ irreducible relatively local
extensions.
\end{defi}

Let $\A$ be a completely rational net on the circle. 
The net $\A$ describes the chiral symmetries of a full two-dimensional CFT. 
For example the Virasoro net with central charge $c<1$ \cite{KaLo2004} is a completely rational net 
and the symmetries it describes are the 
the conformal transformations $\Diff_+(S_1)$ on the circle. For $c\geq 1$ 
the Virasoro net is not completely rational but one can consider larger class of symmetries, 
for example the loop group net 
$\A_{G,k}$ which is known to be completely rational for $G=\SU(N)$ 
and level $k\in\NN$
and which describes $\SU(N)$ gauge transformations.

A \textbf{full CFT based on $\A$} on Minkowski space $\RR^2$, is a maximal local extension $\cB_2$ of the net $\A_2$
which is defined on 
$$\A_2(I_+\times I_-) =\A(I_+)\otimes\A(I_-),\quad I_+\times I_- = \{ (t,x)\in \RR^2: t\pm x\in I_\pm\}\,,$$
where we see $\A$ by restriction as a net on $\RR$.
Since $\A$ is completely rational, $\Rep^I(\A)$ is a unitary modular tensor category $\cC$.
The category of representations of $\A_2$ is equivalent to the category $\cC\boxtimes \bar \cC$
and $\cB_2$ is completely characterized by a commutative Q-system in $\cC\boxtimes \bar \cC$.
With Kawahigashi and Longo, we have gotten a classification of full CFTs in terms of $\A$:
\begin{thm}[\cite{BiKaLo2014}]
\label{thm:Chiral}	
 Full CFTs based on $\A$, \ie maximal local extensions $\cB_2\supset\A_2$ are in one-to-one
correspondence with Morita equivalence classes of non-local extensions $\cB\supset \A$. 
\end{thm}

Given two full CFTs $\cB_2^a,\cB_2^b\supset \A_2$, 
there is a notion of a defect line or phase boundary \cite{BiKaLoRe2014,BiKaLoRe2014-2} between 
the full conformal field theories $\cB_2^a$ and $\cB_2^b$
 on 2D Minkowski space, which is invisible if restricted to $\A_2$,
also called $\A$--topological. If a subfactor arises from $\A$ it comes from such an
	$\A$--topological $\cB^a_2$--$\cB^b_2$ defect. 
\begin{thm}[\cite{BiKaLoRe2014}]
	\label{thm:Defects}
	$\A$--topological $\cB^a_2$--$\cB^b_2$ defects are in one-to-one correspondence with 
sectors $\beta\in \bim {B_a}\cC{B_b}=\langle \beta\prec \iota_b \bim A\cC A \bar\iota_a\rangle$,
for $\bim A\cC A=\Rep^I(\A)$ with  $A=\A(I)$, $B_\bullet=\cB_\bullet(I)$
and $\cB_a,\cB_b\supset \A$ irreducible relatively local
   extension
	corresponding to the full CFT $\cB_2^a,\cB_2^b \supset \A_2$, respectively.
\end{thm}

\begin{rmk}
\begin{enumerate}
  \item The conditions that $B_\bullet$ comes from  relatively local
extension is equivalent to saying that $\bim A\cC A$ 
and $\bim {B_\bullet}\cC {B_\bullet}$ are weakly Morita equivalent in the sense of
  \cite{Mg2003}.
  \item $\bim {B_a}\cC {B_b}$ is a bimodule category 
  over  $\bim {B_a}\cC {B_a}$  and  $\bim {B_b}\cC {B_b}$.  
\item Theorem \ref{thm:Chiral} and \ref{thm:Defects} Show that non-local extensions give 
full CFTs and endomorphisms between these extensions classify 
topological defects between this full CFTs. So if a subfactor arises 
from a conformal net $\A$ in the sense of Definition \ref{defi:Subfactor}, then the subfactor describes a
topological defect. 
\end{enumerate}
\end{rmk}

So far we have seen two sources of UMTCs:
\begin{itemize}
\item Quantum double of subfactors or equivalently Drinfeld centers of unitary fusion categories. 
\item Representation categories of completely rational nets.
\end{itemize}
They for sure don't give the same examples, \eg $\A_{\SU(2)_1}$ has no local extensions and
non-trivial representation category and we have the following characterization of 
nets having quantum doubles as representation category.
\begin{prop} Let $\A$ be a completely rational conformal net. Then the following are equivalent:
\begin{enumerate}
\item
$\Rep(\A)\cong Z(\cF)$ for some unitary fusion category $\cF$.
\item
There exists a holomorphic net $\cB\supset \A$.
\end{enumerate}
Every $\cF$ with $\Rep(\A)\cong Z(\cF)$ gives a particular holomorphic net $\cB_\cF\supset \A$
and there is a Galois correspondence between full subcategories $\cG \subset \cF$ 
and intermediate nets $\cB_F\supset \cB\supset \A$.  
\end{prop}
\begin{proof} If $\Rep(\A)\cong Z(N\subset M)$ then there is an extension, 
such that 
$\A(I)\subset \cB(I)$ is isomorphic to the Longo--Rehren subfactor of $\bim M\cF M$.
Conversely, given $\cB\supset \A$ we get that $\Rep(\A) \cong Z(\cF)$ with 
$\cF=\langle \beta\prec \alpha^+_\rho:\rho\in\Rep(\A)$ the category coming from $\alpha^+$-induction.
See \cite{Bi2015} for details.
\end{proof}
To find a net $\A$ which realizes the quantum double $\cD(N\subset M)$ is the mentioned Problem
\ref{prob:QDnet}. We mention that if we find one net $\A$ with $\Rep(\A)\cong \cD(N\subset M)$ there are infinitely many
examples since for every holomorphic net $\cB$ we have $\Rep(\A)\cong\Rep(\A\otimes \cB)$ and there 
a infinitely man holomorphic nets.

\begin{prop}
  \label{prop:Cofinite} 
  Given a completely rational  net $\A$, then the following are equivalent:
  \begin{enumerate}
    \item  There is a completely rational net 
  $\Rep(\tilde \A)\cong \rev{\Rep(\A)}$
    \item  There is a holomorphic net $\cB$,
  such that $\A\subset \cB$ is normal and cofinite.
  \end{enumerate}
\end{prop}
\begin{proof}
We sketch the proof, more details are in \cite{Bi2015}.
If (1) is true we take $\cB$ the Longo--Rehren extension of $\A\otimes\tilde \A$. Conversely, if (2)
is true, we take $\tilde \A$ to be the coset of $\A\subset \cB$.
\end{proof}

\begin{prop} 
Let $N\subset M$ be a finite index, finite depth subfactor.
\begin{enumerate}
  \item If there is a conformal net with $\Rep(\A)$ braided equivalent to $\cD(N\subset M)$, 
    then $N\subset M$ arises from $\A$. 
Actually it is enough that $\Rep(\A)$ contains a full subcategory braided equivalent to $\cD(N\subset M)$.
  \item Conversely, if $N\subset M$ arises from $\A$, then there is 
  a 2D conformal net $\cB_2\supset \A_2$ with $\Rep(\cB_2)\cong \cD(N\subset M)$.
  \item Further, if $N\subset M$ arises from $\A$, and there is a net $\tilde \A$, with $\Rep(\tilde \A)\cong \rev{\Rep(\A)}$
    then there is a conformal net $\cB\supset \A\otimes \tilde \A$ with $\Rep(\cB)\cong \cD(N\subset M)$.
\end{enumerate}
\end{prop}
\begin{proof} The dual of the Longo--Rehren subfactor applied to $\bim N\cF N$  
gives an extension $\cB$ of $\A$ and it follows that $\bim B \cC B\supset \bim N\cF N$. 
We remark, that in the case  $\Rep(\A)\cong \cD(N\subset M)$ the net  $\cB$ is a holomorphic net. 
We can take an overfactor $\tilde B \supset B$ equivalent to $M\supset N$. This does in general not give a relatively local 
extension of $\cB$ but it gives a relatively local extension of $\tilde \cB\supset \A$ and the inclusion $\iota(B)\subset \tilde B$ does the job.

That $N\subset M$ arises from $\A$ means that there are two extensions $\cB_a,\cB_b\supset $ and for 
$B_\bullet =\cB(I)_\bullet$ there is a morphism $\beta\colon B_a \to B_b$,
such that $\beta(B_a) \subset B_b$ is isomorphic  to $N\subset M$.
But this means that the dual category $\bim {B_a} \cC {B_a}$ contains $\bim[N\subset M] N \cF N$ as a full subcategory.
Since $\Rep(\A_2)\cong Z(\Rep(\A)) \cong Z(\bim {B_a} \cC {B_a})$ it follows from Galois correspondence that there
is a local extension $\cB_2\supset \A_2$ with $\Rep(\cB_2) \cong Z(\bim[N\subset M] N\cF N)\equiv \cD(N\subset M)$.

The net $\A\otimes \tilde \A$ is a conformal net with $\Rep(\A\otimes \tilde \A)\cong Z(\Rep(\A)$, then by exactly the 
 same argument as before, there is a local extension $\cB\supset \A\otimes \tilde \A $ with $\Rep(\cB) \cong  \cD(N\subset M)$.
\end{proof}
In \cite{Bi2015} we used (3) and well-known constructions a to identify net $\A_{N\subset M}$ with $\Rep(\A_{N\subset M})\cong \cD(N\subset M)$ for all 
subfactors with index less than 4. It seems to be interesting to generalize this to other families of subfactors
 and fusion categories. Particular interesting are near group categories \cite{EvGa2014}, since all subfactors in the 
small index classifications 
besides extended Haagerup \cite{Ha1994} seem to be related to near group fusion categories. The double of the 2221 subfactor
or equivalently the $\ZZ_3+3$ near group category is realized by the loop group net $\A_{G_{2,3}}\otimes \A_{\SU(3)_1}$
and the $2^41$ subfactor or the $\ZZ_4+4$ near category is related to a unitary fusion category coming from the  conformal inclusion $\A_{\SU(3)_5}\subset \A_{\SU(6)_1}$ \cf \cite{Li2015}.
This gives hope that near group categories all come from rational nets.

We hope that we convinced the reader 
that the following are interesting problems.
\begin{itemize}
  \item Finding interesting finite index subnets $\A\subset \cB$ for $\cB$ a holomorphic net 
  which give new interesting subfactors/unitary fusion categories
  \item For interesting subfactors, find a completely rational net $\A$ with $\Rep(\A)\cong \cD(N\subset M)$.
    Evans and Gannon argue that for the Haagerup subfactor
such a subnet of the conformal net associated with the $E_8$ lattice \cite{EvGa2011} should exist, but so far it has not been constructed.
    \item Find a general construction for every finite index, finite depth subfactor $N\subset M$ 
which gives a conformal net $\A$ with  $\Rep(\A)\cong \cD(N\subset M)$. This would show that all 
   finite index, finite depth subfactors come from conformal nets.
\end{itemize}

\subsection*{Acknowledgements}
I would like thank Yasuyuki Kawahigashi, Roberto Longo and Zhengwei Liu for discussions.


\def\cprime{$'$}

\end{document}